\documentclass[a4paper,english]{article}

\usepackage[UKenglish]{isodate}
\usepackage{fullpage}
\usepackage{hyperref}

\usepackage{enumitem}

\usepackage{tikz}
\usetikzlibrary{decorations.pathreplacing,arrows}

\usepackage{amsmath,amssymb,oubraces}
\newcommand*{\zz}{\mathbb{Z}}
\newcommand*{\nn}{\mathbb{N}}

\newcommand*{\bw}{\mathbf{w}}

\newcommand*{\cv}{\mathcal{V}}

\newcommand*{\nth}{\text{th}}

\DeclareMathOperator{\occ}{occ}

\newcounter{theorem}[section]

\usepackage{amsthm}
\theoremstyle{plain}
\newtheorem{thm}[theorem]{Theorem}
\newtheorem{lem}[theorem]{Lemma}
\newtheorem{pro}[theorem]{Proposition}
\newtheorem*{pro*}{Proposition}

\newtheorem{cor}[theorem]{Corollary}
\theoremstyle{definition}
\newtheorem{dfn}[theorem]{Definition}

\newtheorem*{exa*}{Example}
\newtheorem*{rmk*}{Remark}

\newcommand{\myqed}{\relax}

\title{Quasiperiods of biinfinite Sturmian words}

\author{Florian
  Barbero\footnote{\textsc{Lirmm}, University of Montpellier, 161 rue
    Ada, $34\,095$ Montpellier Cedex 5, France.} %
  \and %
  Guilhem Gamard\footnote{%
    National Research University Higher School of Economics, Russian Federation. %
    Moscow 125319, Kochnovsky proezd 3.
    This author has been funded by the Russian Academic Excellence Project '5-100'. %
    Corresponding author, e-mail address: \newline
    \href{mailto:guilhem.gamard@normale.fr}{guilhem.gamard@normale.fr}} %
  \and %
  Ana{\"e}l Grandjean\footnote{\textsc{Lacl}, University Paris-Est
    Cr{\'e}teil, 61 avenue du G\'en\'eral de Gaulle, 94010 Cr\'eteil Cedex, France.
  } %
}

\begin{document}

\maketitle

\begin{abstract}
  We study the notion of quasiperiodicity, in the sense of
  ``coverability'', for biinfinite words.  All previous work about
  quasiperiodicity focused on right infinite words, but the passage to
  the biinfinite case could help to prove stronger results about
  quasiperiods of Sturmian words. We demonstrate this by showing that
  all biinfinite Sturmian words have infinitely many quasiperiods,
  which is not quite (but almost) true in the right infinite case, and
  giving a characterization of those quasiperiods.

  The main difference between right infinite and the biinfinite words
  is that, in the latter case, we might have several quasiperiods of
  the same length. This is not possible with right infinite words
  because a quasiperiod has to be a prefix of the word. We study in
  depth the relations between quasiperiods of the same length in a
  given biinfinite quasiperiodic word. This study gives enough
  information to allow to determine the set of quasiperiods of an
  arbitrary word.
\end{abstract}

\section{Introduction}

A finite word $q$ is a \emph{quasiperiod} of a word $w$ if and only if
each position of $w$ is covered by an occurrence of $q$. A word $w$
with a quasiperiod $q \neq w$ is called \emph{quasiperiodic}. For
instance, $abaababaabaaba$ is quasiperiodic and has two quasiperiods:
$aba$ and $abaaba$. Likewise, an infinite word may have several, or
even infinitely many quasiperiods; in the latter case, we call it
\emph{multi-scale quasiperiodic}. The study of quasiperiodicity began
on finite words in the context of text
algorithms~\cite{ApostolicoEhrenfeucht1993Tcs,IliopoulosMouchard1999Jalc},
and was subsequently generalized to right infinite
words~\cite{GlenLeveRichomme2008Tcs,LeveRichomme2007Tcs,Marcus2004Beatcs},
to symbolic dynamical systems~\cite{MarcusMonteil2006Arxiv}, and to
two-dimensional words~\cite{GamardRichomme2015Lata} where it is a
special case of the tiling problem. Finally, a previous
article~\cite{GamardRichomme2016Mfcs} provided a method to determine
the set of quasiperiods of an arbitrary right infinite word. It also
characterized periodic words and standard Sturmian words in terms of
quasiperiods. This is interesting, because periodic words are the
simplest possible infinite words, and Sturmian words are a widely
studied
class~\cite{LeveRichomme2007Tcs,LothaireAlgebraic,MorseHedlundII}
which could be defined as the \emph{least complex non-periodic words}.
These results suggest that quasiperiodicity has some expressive power,
and that the set of quasiperiods is an interesting object to study in
order to get information about infinite words.

The current paper extends to the biinfinite case ($\zz$-words) some
results from~\cite{GamardRichomme2016Mfcs}. The motivations for this
are threefold.

In the two-dimensional case, quasiperiodic $\nn^2$-words and
$\zz^2$-words behave quite
differently~\cite{GamardRichomme2015Lata}. This difference is not
specific to the dimension $2$, so it seems natural to start by
understanding the differences in quasiperiodicity between $\nn$-words
and $\zz$-words.

Quasiperiodicity have been considered not only on infinite words, but
also on subshifts~\cite{MarcusMonteil2006Arxiv}.  However the shift
map does not preserve quasiperiodicity in the right infinite case and
this leads to annoying technicalities.  The biinfinite case is
sometimes considered more natural for subshifts because it turns the
shift map into a bijection. Moreover, it also turns the shift map into
a quasiperiodicity-preserving map, which makes the study of
quasiperiodic subshifts much more convenient.

Finally, a previous article~\cite{GamardRichomme2016Mfcs} gave a
characterization of standard Sturmian words in terms of
quasiperiods. Intuitively, the condition ``standard'' was only needed
because of problems at the origin. By moving to the biinfinite case,
we remove the origin so we can hope for a characterization of all
Sturmian words. (We did not achieve this yet, but it is a possible
continuation of our work.)

\smallskip

The current article makes a first step toward the resolution of these
questions: it generalizes the method to study the set of quasiperiods
of an arbitrary word from~\cite{GamardRichomme2016Mfcs} to the
biinfinite case.  This is not a trivial task because, by contrast with
the right infinite case, we might have several quasiperiods with the
same length. (In the right infinite case, all quasiperiods are
prefixes, thus there may be only one quasiperiod of a given length.)
Therefore we need to determine not only the lengths of the
quasiperiods, but also for each length which factors are quasiperiods
and which are not.

Many natural results about quasiperiodicity on $\nn$-words turned out
to be surprisingly difficult to generalize to $\zz$-words because of
this problem. In addition to show how to determine the set of
quasiperiods of an arbitrary $\zz$-word, we investigate the relations
existing between two quasiperiods of the same length inside a given
biinfinite word. More preciesly, we show that the following conditions
are decidable, given two words $q,r$ of the same length:
\begin{enumerate}[label={(\alph*)}] \itemsep=0pt
\item \label{item:exist}
 there exists a biinfinite word both $q$ and $r$-quasiperiodic;
\item \label{item:infinite}
  each $q$-quasiperiodic biinfinite word contains infinitely many
  occurrences of $r$;
\item \label{item:all}
  each $q$-quasiperiodic biinfinite word is also $r$-quasiperiodic;
\item \label{item:deriv} in any word with quasiperiods $q$ and $r$,
  the derivated sequences of $q$ and $r$ are equal.
\end{enumerate}
Derivated sequences are a tool previously used to build examples and
counter-examples of quasiperiodic words and to show independence
results~\cite{MarcusMonteil2006Arxiv}. A derivated sequence can be
thought as a normal form for quasiperiodic words. Intuitively, when
two derivated sequences are equal, the considered quasiperiods contain
the same information about $\bw$.

Finally, we give a complete description of the set of quasiperiods of
each biinfinite Sturmian word. In particular, we show that each
biinfinite Sturmian word has infinitely many quasiperiods. This
contrasts with the right infinite case, where two Sturmian words of
each slope have no quasiperiods.

\smallskip

\noindent
The paper is structured as follows.

In Section~\ref{sec:det}, we provide a method to study the
quasiperiods of an arbitrary biinfinite word, i.e., a description of
the set of quasiperiods of an arbitrary word.

In Section~\ref{sec:compat}, we define three relations over couples of
words: \emph{compatible}, \emph{definite}, and \emph{positive}. Those
relations are decidable by an algorithm. We show that the couple
$(q,r)$ is compatible if and only if there exists a biinfinite word
$\bw$ having both $q$ and $r$ as quasiperiods (Item~\ref{item:exist}
above). Moreover, the couple $(q,r)$ is definite and positive if and
only if all $q$-quasiperiodic words are also $r$-quasiperiodic
(Item~\ref{item:all}).

In Section~\ref{sec:deriv}, we show that the couple $(q,r)$ is
positive if and only if in any word $\bw$ which is both $q$ and
$r$-quasiperiodic, the derivated sequences along $q$ and $r$ are the
equal (Item~\ref{item:deriv}). We also prove that $(q,r)$ is definite
if and only if each $q$-quasiperiodic word contains infinitely many
copies of $r$ (Item~\ref{item:infinite}).

In Section~\ref{sec:sturm}, we determine the set of quasiperiods of
each biinfinite Sturmian word. In the process we show that all
biinfinite Sturmian words have infinitely many quasiperiods.

Finally in Section~\ref{sec:conclu}, we conclude with a few related
open questions and state our acknowledgements.

\noindent
Figure~\ref{fig:graph} below shows the implications proven in
Sections~\ref{sec:compat} and~\ref{sec:deriv}.

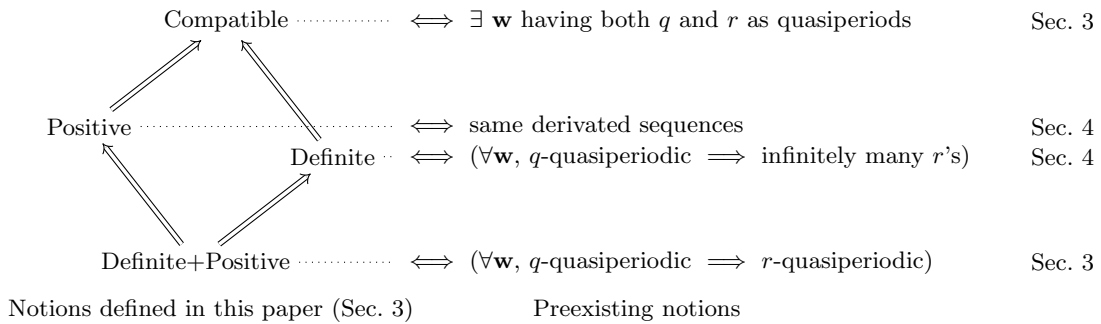
\begin{figure}[hbtp]
  \centering
  \begin{tikzpicture}[scale=0.8] \small
    \node (compat) at(0.25,+2) {Compatible};
    \node (defini) at(+2,-0.25) {Definite};
    \node (positi) at(-2,+0.25) {Positive};
    \node (defpos) at(-0.25,-2) {Definite+Positive};
    \draw[-implies,double equal sign distance] (defini) -- (compat);
    \draw[-implies,double equal sign distance] (positi) -- (compat);
    \draw[-implies,double equal sign distance] (defpos) -- (positi);
    \draw[-implies,double equal sign distance] (defpos) -- (defini);
    \draw (3,+2) node[right] {$\iff$ $\exists$ $\bw$ having both $q$ and $r$ as quasiperiods};
    \draw (3,+0.25) node[right] {$\iff$ same derivated sequences};
    \draw (3,-0.25) node[right] {$\iff$ ($\forall \bw$, $q$-quasiperiodic $\implies$ infinitely many $r$'s)};
    \draw (3,-2) node[right] {$\iff$ ($\forall \bw$, $q$-quasiperiodic $\implies$ $r$-quasiperiodic)};
    \draw[dotted] (compat) -- (3,+2);
    \draw[dotted] (positi) -- (3,+0.25);
    \draw[dotted] (defini) -- (3,-0.25);
    \draw[dotted] (defpos) -- (3,-2);
    \draw (0,-2.8) node{Notions defined in this paper (Sec.~\ref{sec:compat})};
    \draw (7,-2.8) node{Preexisting notions};
    \draw (14,+2) node{Sec.~\ref{sec:compat}};
    \draw (14,+0.25) node{Sec.~\ref{sec:deriv}};
    \draw (14,-0.25) node{Sec.~\ref{sec:deriv}};
    \draw (14,-2) node{Sec.~\ref{sec:compat}};
  \end{tikzpicture}
  \caption{Implications proved in Sections~\ref{sec:compat} and~\ref{sec:deriv}}
  \label{fig:graph}
\end{figure}

\section{Determining the quasiperiods of biinfinite words}
\label{sec:det}

We quickly review classical definitions and notation.  Let $u,v$
denote two finite words and $\bw$ a finite or infinite word.  As
usual, $|u|$ denotes the length of $u$ and $uv$ the concatenation $u$
and $v$.  We note $\bw(i)$ the $i^\nth$ letter of $\bw$; letters are
often considered as words of length $1$. We write $\varepsilon$ for
the empty word.  If $u$ is of length $n$ and satisfies
$u = \bw(i) \bw(i+1) \dots \bw(i+n-1)$, then we say that $u$ is a
\emph{factor} of $\bw$ which \emph{occurs} at position $i$ and which
\emph{covers} positions $i$ to $i+n-1$ (included).  The word $u$ is a
\emph{quasiperiod} of $\bw$ if each position of $\bw$ is covered by an
occurrence of $u$.  In particular, if $\bw$ is finite or right
infinite, then $u$ is a prefix of $\bw$.  If $u$ is a word and
$\alpha, \beta$ two different letters such that $u\alpha$ and $u\beta$
are both factors of $\bw$, we say that $u$ is \emph{right special} in
$\bw$.  Symmetrically, if $\alpha{}u$ and $\beta{}u$ are factors of
$\bw$, then $u$ is \emph{left special} in $\bw$. If $\alpha u \beta$
is a factor of $\bw$, then we say that $u \beta$ is a \emph{successor}
of $\alpha u$, and conversely that $\alpha u$ is a \emph{predecessor}
of $u \beta$ in $\bw$. A word has a unique successor (resp.
predecessor) if and only if it is not right (resp. left)
special. Finally, $|u|_\alpha$ denotes the number of occurrences of
$\alpha$ in $u$.  Unless stated otherwise, all infinite words are
biinfinite, i.e. indexed by $\zz$.

We now have enough vocabulary to state the main theorem
of~\cite{GamardRichomme2016Mfcs}, adapted to the biinfinite case.

\begin{thm}
  \label{thm:old}
  Let $\bw$ denote an infinite word, $q$ a factor of $\bw$ and
  $\alpha$ a letter.
  \begin{enumerate}
  \item Suppose $q$ is a quasiperiod and $q\alpha$ a factor of
    $\bw$. The word $q\alpha$ is a quasiperiod if and only if $q$ is
    \emph{not} right special.
  \item Suppose $q$ is a quasiperiod and $\alpha{}q$ a factor of
    $\bw$. The word $\alpha{}q$ is a quasiperiod if and only if $q$ is
    \emph{not} left special.
  \item Suppose $q\alpha$ is a quasiperiod of $\bw$. The word $q$ is a
    quasiperiod if and only if either $u= q\alpha{}q\alpha$ is not a
    factor of $\bw$, or if $q$ occurs at least $3$ times in $u$.
  \item Suppose $\alpha{}q$ is a quasiperiod of $\bw$. The word $q$ is
    a quasiperiod if and only if either $u= q\alpha{}q\alpha$ is not a
    factor of $\bw$, or if $q$ occurs at least $3$ times in $u$.
  \end{enumerate}
\end{thm}

\noindent
A proof of Theorem~\ref{thm:old} can be found
in~\cite{GamardRichomme2016Mfcs} in the right infinite case; the
adaptation to the biinfinite case is immediate. That theorem basically
states that it is enough to study the set of right special factors and
square factors which are also prefixes to get the set of quasiperiods
of a given right infinite word. As special and square factors are
well-understood in combinatorics on words, it generally little
additional word to get the set of quasiperiods of a given
right infinite word. We will comment on the biinfinite version of the
theorem, which we just stated, in a few paragraphs.

We can extend this theorem a bit further, but to do so we need the
notion of \emph{overlap}.

\begin{dfn}
  Let $q$ denote a finite word. An \emph{overlap of $q$} is a word $w$
  having $q$ as a prefix and as a suffix, such that
  $|q| < |w| \leq 2|q|$. More generally, a \emph{$k$-overlap of $q$}
  is a word of the form $uv$, where $u$ is a $(k-1)$-overlap and $v$
  is such that $qv$ is an overlap of $q$.

  The quantity $2|q| - |w|$ is called the \emph{span} of the overlap.
  If $q$ is fixed, then an overlap is uniquely determined by its span,
  thus we note $\cv_q(m)$ the overlap of $q$ having span $m$ (if it
  exists).  We write $\cv_q(n_1, n_2, \dots, n_{k-1})$ the $k$-overlap
  built from overlaps $\cv_q(n_1)$, $\cv_q(n_2)$, etc.\ and we call
  $n_i$ the \emph{$i^{\mbox{\small th}}$ span} of this overlap.
\end{dfn}

An \emph{overlap} (without any explicit $k$) is thus a $2$-overlap. An
infinite word $\bw$ is $q$-quasiperiodic if and only if two
consecutive occurrences of $q$ in $\bw$ always form an overlap.

In general, we might have more than two occurrences of $q$ in an
overlap of $q$. For instance, $\cv_{aaa}(1) = aaaaa$ contains $3$
occurrences of $aaa$. We say that $w$ is a \emph{proper $k$-overlap of
  $q$} if $w$ is a $q$-quasiperiodic word containing exactly $k$
occurrences of $q$. We write $\cv^*_q(n_1, \dots, n_{k-1})$ when we
mean that $\cv_q(n_1, \dots, n_{k-1})$ is a \emph{proper} $k$-overlap
of $q$. A \emph{proper overlap} is implicitly a \emph{proper
  $2$-overlap}.

\begin{lem}
  \label{lem:patrice}
  Let $u$ denote a word and $\alpha,\beta$ letters. If $u\beta$ is a
  factor of an overlap of $u\alpha$, then $\alpha=\beta$.
\end{lem}
\begin{proof}
  Let $w$ denote an overlap of $u\alpha$; by definition of an overlap,
  there exist words $p$, $s$ (possibly empty) such that $u=ps$ and
  $w = u\alpha{}s\alpha = ps\alpha s\alpha$.  If $u\beta$ is a factor
  of $w$, then $s\beta$ is a factor of $s\alpha{}s\alpha$. Let $x,y$
  denote the words such that $s\alpha{}s\alpha=xs\beta{}y$. Observe
  that $|xy|=|s\alpha|$, that $x$ is a prefix and $y$ a suffix of
  $s\alpha$ to conclude that $xy=s\alpha$. Thus we can simplify
  $|s\alpha{}s\alpha|_\alpha=|xs\beta{}y|_\alpha$ into
  $|s\alpha|_\alpha = |s\beta|_\alpha$, which implies
  $|\alpha|_\alpha=|\beta|_\alpha$ and $\alpha=\beta$.  \myqed
\end{proof}

\begin{pro}
  \label{pro:predsuc}
  Let $\bw$ denote an infinite word, and $q$ a quasiperiod of length
  $n$ of $\bw$.  A successor of $q$ is a quasiperiod of $\bw$ if and
  only if $q$ is not right special.  A predecessor of $q$ is a
  quasiperiod of $\bw$ if and only if $q$ is not left special.
\end{pro}
\begin{proof}
  Let $\alpha$, $\beta$ denote letters and $u$ denote a word such that
  $\alpha u$ is a quasiperiod and $\alpha u \beta$ a factor of $\bw$.
  If $u\beta$ is a quasiperiod of $\bw$ and $u \gamma$ is also factor
  of $\bw$ for a letter $\gamma \neq \beta$, then $u\gamma$ is a
  factor of an overlap of $u \beta$.  Lemma~\ref{lem:patrice} shows
  that $\beta=\gamma$: a contradiction.  Conversely if $\alpha u$ is
  not right special, then every occurrence of $\alpha u$ continues
  into an occurrence of $u \beta$; since $\alpha u$ covers $\bw$, so
  does $u \beta$. The left special case is symmetric.  \myqed
\end{proof}

Theorem~\ref{thm:old} and Proposition~\ref{pro:predsuc} together imply
that, in order to understand the set of quasiperiods of a biinfinite
word, it is enough to know its set of special factors and its set of
square factors. These two types of factors are already well-studied
and well-understood in combinatorics on words, therefore we can reuse
this knowledge when we need to get the set of quasiperiods of an
infinite word.

Proposition~\ref{pro:predsuc} has another interesting consequence: if
an infinite, aperiodic word $\bw$ has a quasiperiod of some length
$n$, then it also has a left-special quasiperiod $\ell$ and a right
special quasiperiod $r$ of length $n$.  More precisely, the set of
quasiperiods of some length $n$ is given by a union of chains of the
form $\{u_1, \dots, u_k\}$, where $u_1$ is left special, $u_k$ is
right special, no other $u_i$ is special, and $u_{i+1}$ is the
(unique) successor of $u_i$ for each $1 \leq i < k$. If $q$ belongs to
such a chain, we call $u_1$ its left-special predecessor and $u_k$ its
right special successor.

After working out several examples, one may conjecture that there is
at most one right special (and thus one left special) quasiperiod of a
given length in any biinfinite word. In this case, there would be at
most one chain of quasiperiods of a given length, so it would be easy
to determine the set of quasiperiods of an arbitrary biinfinite
word. Unfortunately the following example disproves this
conjecture. Let $q = aba\,ab\,aba$, $r = aba\,ba\,aba$ and $\bw$ be
defined by:
\begin{equation}
  \label{eq:badex} %
  \bw
  = {}^\omega (a^{-1}q) \cdot (q)^\omega
  =
  \dots \,
  \overunderbraces%
  {&&\br{2}{r}&&\br{3}{r}&\br{1}{r}&\br{1}{r}&&&&&}%
  {&
  ba ab &aba\,
  ba ab &a&ba\,
  ba &ab &a&ba\,\cdot\,
  aba& ab aba\,
  aba& ab aba\,
  aba& ab aba
  &}%
  {&&&\br{4}{r}&&&&&&&&}
  \, \dots
\end{equation}
where the end of each occurrence of $q$ in $\bw$ is showed by a space.
The definition of $\bw$ makes it clear that $q$ is a quasiperiod of
$\bw$.  As the excerpt of $\bw$ suggests, $r$ is also a quasiperiod of
$\bw$: since the word is ultimately periodic, the same behaviour
repeats to the left and to the right. It can be directly observed in
the excerpt that both $q$ and $r$ are right special. This example is
the simplest ``pathological case'' which we mentioned in the
introduction.

\section{Checking implcations between two quasiperiods}
\label{sec:compat}

In this section we show that it is decidable to check, given two
finite words $q$ and $r$ of the same length, which of the following is
true:
\begin{enumerate}
\item Any $q$-quasiperiodic biinfinite word is also $r$-quasiperiodic;
\item there exists an infinite word which is $q$- and
  $r$-quasiperiodic, and another one which is just $q$-quasiperiodic;
\item no infinite word may have both quasiperiods $q$ and $r$ at the
  same time.
\end{enumerate}
First we develop a bit of vocabulary to state the conditions in a
convenient way.

\begin{lem}
  \label{lem:no-qrrq}
  Let $q,r$ denote two different words of the same length and $w$ a
  proper overlap of $q$. The word $w$ has at most one occurrence of
  $r$.
\end{lem}
\begin{proof}
  By a classical lemma~\cite[Prop.~1.3.4]{Lothaire1983}, there exist
  finite words $x,y$ and an integer $k$ satisfying $q = (xy)^kx$ and
  $w = (xy)^{k+1}x$. Moreover, $xy$ is a primitive word. If it were
  not, call $z$ its primitive root and observe that an occurrence of
  $q$ would start at position $|z|$ in $w$, yielding three occurrences
  of $q$ in $w$, a contradiction. Additionally, we have either
  $k \geq 1$ or $y = \varepsilon$. Indeed, if $k = 0$ and $|y|\geq 1$,
  we would have $q=x$ and $w=xyx$, implying $|w|>2|q|$, a
  contradiction with the definition of an overlap. We treat the cases
  $k \geq 1$ and $y = \varepsilon$ separately.

  First, suppose $k \geq 1$. As $|q|=|r|$, all occurrences of $r$ in
  $w$ must start at positions between $1$ and $|xy|$ (included). Call
  $u$ the prefix of length $|xy|$ of $r$. The word $u$ is a factor of
  $xyxy$. Because $xy$ is primitive, each factor of length $|xy|$
  occurs only once in $xyxy$, except $xy$
  itself~\cite[Prop.~1.3.2]{Lothaire1983}. This means that there can
  only be one occurrence of $u$, and therefore of $r$, starting in the
  first $|xy|$ letters of $w$.

  Now suppose $k = 0$. By the previous remarks, this implies
  $y = \varepsilon$, thus $q=x$ and $q$ is primitive. As a
  consequence, each factor of length $|q|$ in $qq = w$ occurs only
  once, excepted $q$ itself (otherwise, $q = q_1q_2 = q_2q_1$ for some
  finite words $q_1, q_2$, and~\cite[Proposition~1.3.2]{Lothaire1983}
  contradicts primitivity). In particular $r$, if it occurs at all,
  occurs only once.
\myqed
\end{proof}

\begin{dfn}
  Let $q,r$ denote finite nonempty words of the same length and $m,n$
  natural integers. If the proper overlap $\cv^*_q(m)$ exists and
  contains $r$ as a factor, then we write $\occ(q,r,m)$ for the
  position of $r$ in $\cv^*_q(m)$; otherwise $\occ(q,r,m)$ is not
  defined. (Lemma~\ref{lem:no-qrrq} ensures that if $\occ(q,r,m)$
  exists, then it is unique.) If both $\occ(q,r,m)$ and $\occ(q,r,n)$
  exist, then we define the quantity
  \begin{equation}
    \label{eq:def-f}
    f_{q,r}(m,n) = m + \occ(q,r,m) - \occ(q,r,n)
  \end{equation}
  otherwise, $f_{q,r}(m,n)$ is undefined.
\end{dfn}

We insist on the fact that $\occ(q,r,m)$ is defined only where
$\cv^*_q(m)$ is defined and contains an occurrence of $r$. If
$\cv_q(m)$ is not a \emph{proper} overlap (i.e. it contains more than
two occurrences of $q$, like $\cv_{aaa}(1)$), then $\occ(q,r,m)$ is
not defined. The quantity $f_{q,r}(m,n)$ is defined if and only if
both $\occ(q,r,m)$ and $\occ(q,r,n)$ are. Moreover, $q$ and $r$ are
not symmetric: $f_{q,r} \neq f_{r,q}$.

Here is the intuitive interpretation of $f_{q,r}$. Let $m,n$ denote
natural integers such that $w=\cv^*_{q}(m,n)$ is a proper $3$-overlap
of $q$. By Lemma~\ref{lem:no-qrrq}, the word $w$ has at most $2$
occurrences of $r$. Suppose it has exactly two. If these two
occurrences form an overlap of $r$, then $f_{q,r}(m,n)$ is the span of
this overlap. If these two occurrences do not overlap, then there
exists a nonempty word $s$ such that $rsr$ is a factor of $w$; in this
case, $f_{q,r}(m,n)=-|s|$. If $w$ has less than two occurrences of
$q$, then $f_{q,r}(m,n)$ is not defined.

\begin{exa*}
  In Equation~\eqref{eq:badex} we had $q = aba\,ab\,aba$ and
  $r = aba\,ba\,aba$; in this case the function $f_{q,r}$ is given by:
  \begin{center}
    \begin{tabular}{l | l l l }
      $f$ & 0 &  1 & 3 \\ \hline
      0   & 0 & -2 & 0 \\
      1   & 3 &  1 & 3 \\
      3   & 3 &  1 & 3  
    \end{tabular}
  \end{center}
\end{exa*}

Computing $f_{q,r}$ given two finite words $q$ and $r$ of the same
length can be done in $O(|q|^3)$ time. For each $m$ and for each $n$
between $0$ and $|q|$ (included), compute $\cv^*_q(m)$ and
$\cv^*_q(n)$; in each of them, test whether $r$ appears as a factor;
if so, use Equation~\eqref{eq:def-f} to compute the value of
$f(m,n)$. Otherwise, $f(m,n)$ is not defined. The computation of
$\cv^*_q(m)$ and $\cv^*_q(m)$, and the search for $r$, can be done in
$O(|q|)$ time using an optimal string-searching algorithm.

\begin{lem}
  \label{lem:lines}
  \label{lem:gal-nfo}
  Let $q$, $r$ be finite words and $\{a_1, \dots, a_k\}$ the set of
  integers such that the proper overlap $\cv^*_q(a_i)$ exists and
  contains one occurrence of $r$. Then, for all $s_1, \dots, s_n$ in
  $\{a_1, \dots, a_k\}$, the following equation holds:
  \begin{equation}
    \label{eq:sum}
    s_1 + \dots + s_k = f(s_1, s_2) + \dots + f(s_k, s_1).
  \end{equation}
  In particular, for all integers $k,l,m,n$ in $\{a_1, \dots, a_k\}$ we have:
    $f(m,m)=m$;
    the relation $f(m,n)=m$ implies that $f(n,m)=n$; and
    the relation $f(m,k)=f(m,l)$ implies that $f(n,k)=f(n,l)$.
\end{lem}
\begin{proof}
  Since $\cv^*_q(a_i)$ contains exactly one occurrence of $r$ for all
  $i$, Lemma~\ref{lem:no-qrrq} implies that $\cv^*_q(a_i, a_j)$
  contains exactly two occurrences of $r$ for all $i,j$. As a
  consequence, $f_{q,r}(a_i, a_j)$ is always defined. Conversely,
  $f_{q,r}(m,n)$ is not defined if
  $\{m,n\}\not\subseteq\{a_1, \dots, a_k\}$, since $\cv^*_q(m,n)$ does
  not contain two occurrences of $r$.

  For Equation~\eqref{eq:sum}, first we compute:
  \begin{equation*}
    f(a_1, a_2) + f(a_2,a_1) = a_1 + \occ(a_1) - \occ(a_2) + a_2 + \occ(a_2) - \occ(a_1)
    = a_1 + a_2
  \end{equation*}
  and then the result is easily proved by induction. The three other
  facts of this lemma are immediate consequences of
  Equation~\eqref{eq:sum} and of the definition of $f$.  \myqed
\end{proof}


Now we have enough machinery to state conditions on $(q,r)$ which
characterize situations where $q$-quasiperiodicity implies
$r$-quasiperiodicity, or implies non-$r$-quasiperiodicity.

\begin{dfn}
  Let $q,r$ denote finite nonempty words of the same length. The
  couple $(q,r)$ is:
  \begin{itemize}
  \item \emph{compatible} if there exist integers $m,n$ such that
    $f_{q,r}(m,n)$ is defined;
  \item \emph{definite} if $f_{q,r}(m,n)$ is defined wherever
    $\cv^*_q(m,n)$ is;
  \item \emph{positive} if $f_{q,r}(m,n)$ is defined at least on one
    couple and is nonnegative wherever it is defined.
  \end{itemize}
\end{dfn}

\noindent
Since $f_{q,r}$ is computable in time $O(|q|^3)$, those relations are
testable with the same time complexity.

\begin{thm}
\label{thm:rel-qp}
Let $q,r$ denote two finite, nonempty words of the same length.
\begin{enumerate}
\item The couple $(q,r)$ is non-compatible if and only if
  $q$-quasiperiodicity implies non-$r$-quasiperiodicity.
\item The couple $(q,r)$ is definite and positive if and only if
  $q$-quasiperiodicity implies $r$-quasiperiodicity.
\item The couple $(q,r)$ is compatible, but not definite positive if
  and only if there exists a biinfinite word with quasiperiods $q$ and
  $r$, and another biinfinite word with only quasiperiod $q$.
\end{enumerate}
\end{thm}
\begin{proof} We prove the three statements separately.
  
  \emph{Statement 1.} Let $\bw$ denote $q$-quasiperiodic word. It
  contains a factor of the form $u=\cv^*_q(m, n)$. By hypothesis
  $f_{q,r}(m, n)$ is not defined, which means that $u$ contains either
  $0$ or $1$ occurrences of $r$. By Lemma~\ref{lem:no-qrrq}, at least
  one position in $\bw$ is not covered by $r$, so, $\bw$ is not
  $r$-quasiperiodic.
  
  Conversely, suppose that each $q$-quasiperiodic word is
  non-$r$-quasiperiodic. Let $m,n$ denote a pair of integers such that
  the proper $3$-overlap $\cv^*_q(m,n)$ exists and consider the
  infinite periodic word given by
  $\bw = \cv_q(\dots m, n, m, n, m, n \dots)$.  Since $\bw$ is not
  $r$-quasiperiodic, either $\cv^*_q(m,n)$ or $\cv^*_q(n,m)$ (or both)
  contains less than two occurrences of $r$. In other terms, either
  $\cv^*_q(m,n)$ or $\cv^*_q(n,m)$ is not defined, and by
  Lemma~\ref{lem:lines} the other one is not defined either.  Since
  this reasonning holds for any $m,n$ where $\cv^*_q(m,n)$ exists, the
  function $f_{q,r}$ is nowhere defined.

  \smallskip

  \emph{Statement 2.} Suppose $(q,r)$ is definite and positive and
  consider $\bw$ a $q$-quasiperiodic biinfinite word. Any position in
  $\bw$ is covered by an occurrence of $q$; let $m,n$ denote the
  integers such that this occurrence is the middle one in the proper
  $3$-overlap $\cv^*_q(m.n)$. By hypothesis, $f_{q,r}(m,n)$ is defined
  and positive, so $\cv^*_q(m,n)$ contains a proper overlap of
  $r$. Lemma~\ref{lem:no-qrrq} implies that this proper overlap of $r$
  covers the middle occurrence of $q$. Consequently, any position in
  $\bw$ is covered by an occurrence of $r$.

  Conversely, suppose that $q$-quasiperiodicity implies
  $r$-qua\-si\-pe\-rio\-di\-ci\-ty. Let $m,n$ denote an arbitrary pair
  of integers $m,n$ such that the proper $3$-overlap $\cv^*_q(m,n)$
  exists and $\bw$ denote the periodic biinfinite word given by
  $\bw=\cv^*_q(\dots m,n,m,n,m,n \dots)$. By hypothesis this word is
  $r$-quasiperiodic, so by Lemma~\ref{lem:no-qrrq} the word
  $\cv^*_q(m,n)$ contains a proper overlap of $r$. Consequently,
  $f_{q,r}(m,n)$ is defined and positive.

  \smallskip
  
  \emph{Statement 3.} The proof is immediate as this statement
  exhausts all possibilities not covered by Statements~1 and 2.

  \myqed
\end{proof}

\section{On compatible and positive couples of quasiperiods}
\label{sec:deriv}

In this section, we investigate what the property ``compatible and
positive'' implies for a couple of words (not necessarily
definite). We get a characterization in terms of derivated sequences,
and another one in terms of chains of quasiperiods.

\smallskip

The concept derivated sequence originates from Mouchard's work on
quasiperiodic finite words~\cite{IliopoulosMouchard1999Jalc}, and was
later used by Marcus and Monteil to establish independence results
between quasiperiodicity and other properties on right infinite
words~\cite{MarcusMonteil2006Arxiv}. We start by recalling the
definition.

\begin{dfn}
Let $\bw$ denote a biinfinite word and $q$ one of its
quasiperiods. The \emph{sequence of positions} of $q$ in $\bw$ is the
sequence $(q_n)_{n \in \zz}$ of positions of occurrences of $q$ in
$\bw$, in increasing order, such that $q_0$ is the position of the
leftmost occurrence covering the position $0$. If $(q_n)_{n \in \zz}$
is the sequence of positions of $q$ in $\bw$, then
$(q_{n+1} - q_n)_{n \in \zz}$ is called the \emph{derivated sequence}
of $\bw$ along $q$.
\end{dfn}

\noindent
For example, in Equation~\eqref{eq:badex}, the derivated sequence of
$\bw$ along $q$ is ${}^{\omega}(7)(8)^\omega$ and the derivated
sequence along $r$ is ${}^{\omega}(7) \, 5 \, (8)^\omega$.  Observe
that a word is $q$-quasiperiodic if and only if its derivated sequence
along $q$ is bounded by $|q|$. In this case, the derivated sequence
contains enough information to reconstruct the initial word.

\medskip

Chains of quasiperiods were already mentioned in
Section~\ref{sec:det}. Recall the following consequence of
Theorem~\ref{thm:old} and Proposition~\ref{pro:predsuc}: in an
infinite word $\bw$, the set of quasiperiods of some length $n$ is
given by a union of chains of the form $\{u_1, \dots, u_k\}$, where
$u_1$ is left special, $u_k$ is right special, no other $u_i$ is
special, and $u_{i+1}$ is the (unique) successor of $u_i$ for each
$1 \leq i < k$.  Equation~\eqref{eq:badex} shows an example of a word
having two such chains for length $8$.

\begin{thm}
  \label{thm:deriv}
  Let $\bw$ denote a biinfinite word and $q$, $r$ denote two
  quasiperiods of $\bw$ of the same length. The following statements
  are equivalent:
  \begin{enumerate}
  \item the couple $(q,r)$ is compatible and positive;
  \item for all word $\bw$ having quasiperiods $q$ and $r$, the
    derivated sequences along $q$ and $r$ are equal;
  \item for all word $\bw$ having quasiperiods $q$ and $r$, those
    quasiperiods belong to the same chain.
  \end{enumerate}
\end{thm}

\noindent
We actually prove something slightly more precise: the next
proposition implies Theorem~\ref{thm:deriv}.

\begin{pro}
  \label{pro:spec}
  Let $\bw$ denote a biinfinite word and $q$, $r$ denote two
  quasiperiods of $\bw$ of the same length. The following statements
  are equivalent:
  \begin{enumerate}
  \item for each integers $m,n$ such that the proper $3$-overlap
    $\cv^*_q(m,n)$ exists and is a factor of $\bw$, we have
    $f_{q,r}(m,n) \geq 0$.
  \item the derivated sequences of $\bw$ along $q$ and along $r$ are
    equal up to a shift of one position;
  \item the quasiperiods $q$ and $r$ belong to the same chain in
    $\bw$.
  \end{enumerate}
\end{pro}

%

The next lemma gives $2 \implies 1$ in Proposition~\ref{pro:spec},
because $x$ and $y$ are nonnegative integers. However it is actually
more general and we will also reuse it later in the proof.

\begin{lem}
  \label{lem:equiv-f} %
  Let $\bw$ denote an infinite word and $q$, $r$ two quasiperiods of
  $\bw$ of the same length.  The derivated sequences of $q$ and $r$ in
  $\bw$ are the same if and only if either: for each pair of natural
  integers $(x,y)$ such that $f(x,y)$ is defined, we have
  $f(x, y) = x$; or for each such pair $(x,y)$, we have $f(x,y) = y$.
\end{lem}
\begin{proof}
  Call $(q_n)_{n \in \zz}$ the sequence of positions of $q$ in $\bw$,
  and similarly $(r_n)_{n \in \zz}$ the sequence of positions of $r$
  in $\bw$; observe that
  $r_{n+1} - r_n = f(q_n - q_{n-1}, q_{n+1} - q_n)$.
  The fact that $f(x,y) = x$ and $f(x,y) = y$ respectively translate to
  \begin{equation*}
    f(q_{n+1}-q_n, q_n-q_{n-1}) = q_{n+1} - q_n, \quad \text{ and } \quad
    f(q_{n+1}-q_n, q_n-q_{n-1}) = q_n - q_{n-1}.
  \end{equation*}
  By replacing in the previous equation, we get the two possibilities
  \begin{equation*}
    r_{n+1} - r_n = q_{n+1} - q_n, \quad \text{ and } \quad
    r_{n+1} - r_n = q_n - q_{n-1},
  \end{equation*}
  which both imply that the derivated sequences are equal up to a
  shift. The converse argument work symmetrically.
\myqed
\end{proof}

The next lemma proves $1 \implies 2$ in Proposition~\ref{pro:spec},
which is the most technical part.

\begin{lem}
  \label{lem:1imp2} %
  Let $\bw$ denote an infinite word and $q$, $r$ two quasiperiods of
  $\bw$ of the same length. Suppose that for each pair of integers
  $(m,n)$ such that the proper $3$-overlap $\cv^*_q(m,n)$ exists and
  is a factor of $\bw$, we have $f_{q,r}(m,n) \geq 0$. Then the
  derivated sequences of $\bw$ along $q$ and along $r$ are identical,
  up to a shift of one position.
\end{lem}
\begin{proof}
  Let $\tau_1, \tau_2, \dots, \tau_m$ denote all the integers such
  that the proper overlap $\cv^*_q(\tau_i)$ exists and is a factor of
  $\bw$; sort the $\tau_i$ by increasing length. If $x,y$ are integers
  such that the proper $3$-overlap $\cv^*_q(x,y)$ exists and is a
  factor of $\bw$, then we call $(x,y)$ an \emph{occurring couple}.

  We only need to prove that either for each $x$ such that
  $(\tau_1, x)$ is an occurring couple, we have
  $f(\tau_1,x) = \tau_1$; or that for each such $x$, we have
  $f(x,\tau_1)=\tau_1$. Indeed, if $f(\tau_1,x)=\tau_1$ for all $x$
  (or the opposite one), then for each occurring couple $(x,y)$ we
  have $f(\tau_1,x)=f(\tau_1,y)$; by Lemma~\ref{lem:lines} we deduce
  that $f(x,y)=f(x,x)=x$; subsequently Lemma~\ref{lem:equiv-f} shows
  that this is sufficient to finish our proof. Therefore now we argue
  that $f(\tau_1, x) = \tau_1$ for each $x$ such that $(\tau_1, x)$ is
  an occurring couple.

  Let $\sigma_1 = \tau_1$ and $\sigma_2, \dots, \sigma_n$ all the
  integers such that $(\sigma_1, \sigma_i)$ is an occurring couple;
  sort the $\sigma_i$ by increasing length (in particular
  $(\sigma_i)_{1 \leq i \leq n}$ is a subsequence of
  $(\tau_i)_{1 \leq i \leq m}$). The couple $(\sigma_1, \sigma_1)$ is
  not necessarily an occurring couple, but Lemma~\ref{lem:lines}
  guarantees that $f(\sigma_1, \sigma_1)$ is well-defined and that
  $f(\sigma_1, \sigma_1) = \sigma_1$. We can assume that
  $f(\sigma_1, \sigma_2) = \sigma_1$; if it is not the case, then
  $f(\sigma_2, \sigma_1) = \sigma_1$ by Lemma~\ref{lem:lines} and
  without loss of generality we consider the function
  $f'(x,y) = f(y,x)$ instead of $f$. Now reason by contradiction and
  consider the smallest integer $j$ such that
  $f(\sigma_1, \sigma_j) \neq \sigma_1$. By hypothesis we can rule out
  $f(\sigma_1, \sigma_j) < 0$, so it remains three cases to analyse.
  
  \textbf{Case 1.} If $f(\sigma_1, \sigma_j) > \sigma_j$, then
  use Lemma~\ref{lem:lines} to write
  $\sigma_1 + \sigma_j = f(\sigma_1, \sigma_j) + f(\sigma_j, \sigma_1)$, which
  is equivalent to $f(\sigma_j,\sigma_1)=\sigma_1+\sigma_j-f(\sigma_1,\sigma_j)$.
  The quantity $\sigma_j-f(\sigma_1,\sigma_j)$ is negative so we have
  $f(\sigma_j, \sigma_1) < \sigma_1$, which is a contradiction since
  $\sigma_1$ is the smallest possible span.

  \textbf{Case 2.} If $\sigma_1 < f(\sigma_1, \sigma_j) = x < \sigma_j$,
  then consider the $q$-quasiperiodic word whose derivated sequence is
  ${}^\omega(x\, \sigma_1\, \sigma_j)^\omega$; it would have
  $f(\sigma_1,x) = \sigma_1$ and $f(x,\sigma_1) = x$. Thus
  $f_{q,r}(x,x)=\sigma_1$, but Lemma~\ref{lem:gal-nfo} implies $f_{q,r}(x,x)=x$:
  we have a contradiction.
  
  \textbf{Case 3.} Finally, suppose we have
  $f(\sigma_1, \sigma_j) = \sigma_j$. Recall that $j$ is minimal and that $j > 2$.
  By Lemma~\ref{lem:gal-nfo} we have
  $f(\sigma_j, \sigma_1) = \sigma_1$, and by Lemma~\ref{lem:lines}
  the relations $f(\sigma_1,\sigma_1) = f(\sigma_1,\sigma_2)$ and
  $f(\sigma_j, \sigma_1) = \sigma_1$ imply that
  $f(\sigma_j, \sigma_2) = \sigma_1$ as well. Therefore we have four
  $3$-overlaps of $q$ with spans $(\sigma_1, \sigma_1)$;
  $(\sigma_1, \sigma_2)$; $(\sigma_j, \sigma_1)$; and
  $(\sigma_j, \sigma_2)$; all with the same induced $r$-overlap, which
  has span $\sigma_1$.

  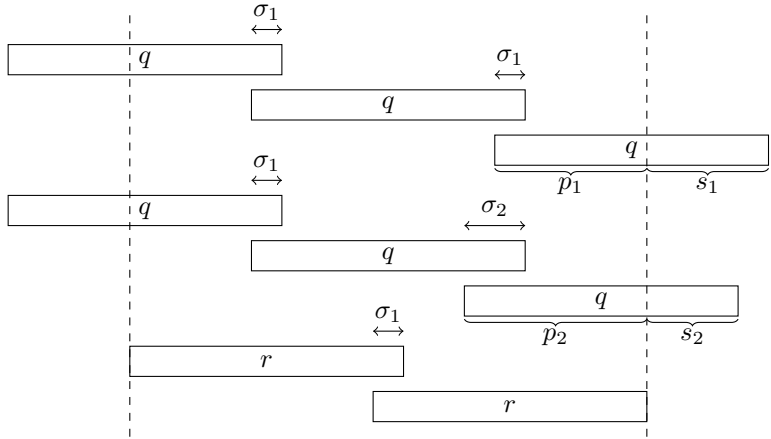
\begin{figure}[htbp]
    \centering
    \begin{tikzpicture}[scale=0.4]
      \draw (-8, 8  ) rectangle node{$q$} ++(9,1);
      \draw ( 0, 6.5) rectangle node{$q$} ++(9,1);
      \draw ( 8, 5  ) rectangle node{$q$} ++(9,1);

      \draw (-8, 3  ) rectangle node{$q$} ++(9,1);
      \draw ( 0, 1.5) rectangle node{$q$} ++(9,1);
      \draw ( 7, 0  ) rectangle node{$q$} ++(9,1);

      \draw (-4, -2  ) rectangle node{$r$} ++(9,1);
      \draw ( 4, -3.5) rectangle node{$r$} ++(9,1);

      \draw[dashed] (-4,-4) -- (-4,10);
      \draw[dashed] (13,-4) -- (13,10);

      \draw[<->] (+4,-0.5) -- node[above]{$\sigma_1$} ++(1,0);
      \draw[<->] ( 0, 4.5) -- node[above]{$\sigma_1$} ++(1,0);
      \draw[<->] ( 0, 9.5) -- node[above]{$\sigma_1$} ++(1,0);
      \draw[<->] ( 7, 3  ) -- node[above]{$\sigma_2$} ++(2,0);
      \draw[<->] ( 8, 8  ) -- node[above]{$\sigma_1$} ++(1,0);

      \draw[decorate,decoration={brace,mirror},yshift=-1mm] (8,5) -- node[below]{$p_1$} ++(5,0);
      \draw[decorate,decoration={brace,mirror},yshift=-1mm] (13,5) -- node[below]{$s_1$} ++(4,0);
      
      \draw[decorate,decoration={brace,mirror},yshift=-1mm] (7,0) -- node[below]{$p_2$} ++(6,0);
      \draw[decorate,decoration={brace,mirror},yshift=-1mm] (13,0) -- node[below]{$s_2$} ++(3,0);
    \end{tikzpicture}
    \caption[Illustration of the proof of
    Lemma~\ref{lem:1imp2}]{Illustration of Case~3 in the proof of
      Lemma~\ref{lem:1imp2}}
    \label{fig:main-fat}
  \end{figure}

  From $f(\sigma_1,\sigma_1) = f(\sigma_1, \sigma_2) = \sigma_1$ we
  deduce that $\cv_q(\sigma_1,\sigma_1)$ and $\cv_q(\sigma_1,\sigma_2)$ have a
  common factor $\cv_r(\sigma_1)$. In either case, the middle occurrence
  of $q$ is contained in this factor and last occurrence of $q$ starts
  in this factor. This situation is displayed on
  Figure~\ref{fig:main-fat}. There exist words $s_1, s_2$ such that
  $Q(\sigma_1)$ is a suffix of $\cv_r(\sigma_1)s_1$, and $\cv_q(\sigma_2)$ is
  a suffix of $\cv_r(\sigma_1)s_2$. Call $p_1, p_2$ the words satisfying
  $p_1s_1 = p_2s_2 = q$, and without loss of generality suppose that
  $|p_2|>|p_1|$. Observe that both $p_1$ and $p_2$ are suffixes of the
  same word $\cv_r(\sigma_1)$, and $|p_1|=|p_2|+\sigma_2-\sigma_1$.  As
  $p_1,p_2$ are also both prefixes of $q$, we deduce that $p_2$ has a
  period $\sigma_2 - \sigma_1$. From
  $f(\sigma_1, \sigma_2) = f(\sigma_j, \sigma_2) = \sigma_1$ the same
  argument proves that there exist words $p'_1, p'_2$ such that
  $\cv_q(\sigma_1)$ is a prefix of $p'_1\cv_r(\sigma_1)$, and $\cv_q(\sigma_j)$
  is
  a prefix of $p'_2\cv_r(\sigma_1)$. Call $s'_1, s'_2$ the words
  satisfying $p'_1s'_1 = p'_2s'_2 = q$, and without loss of generality
  suppose that $|s'_2| > |s'_1|$. Then remark that $s'_2$ has a period
  $\sigma_j - \sigma_1$. In order to simplify notation in the rest of
  the proof, let $\alpha = p_2$ and $\beta = s'_1$.

  We have
  $|\cv_r(\sigma_1)|=2|q|-\sigma_1=|q|+|\alpha|+|\beta|-\sigma_2-\sigma_j$.
  Therefore $|q|+\sigma_2+\sigma_j - \sigma_1 = |\alpha|+|\beta|$.
  Since $\alpha$ is a prefix and $\beta$ a suffix of $q$, by a length
  argument $\alpha$ has a non-empty suffix which is a prefix of
  $\beta$; call it $\theta$. We have
  $|\theta| = |\alpha|+|\beta|-|q| = \sigma_j+\sigma_2-\sigma_1 \geq
  (\sigma_j - \sigma_1) + (\sigma_2 - \sigma_1)$. By the Fine-Wilf
  Theorem (see~\cite[Prop.~1.3.5 and its proof]{Lothaire1983}),
  $\theta$ has a period $\delta$ of length
  $|\delta| = \gcd(\sigma_j-\sigma_1, \sigma_2-\sigma_1)$. A period of
  $\alpha$ is a suffix of $\theta$ and a period of $\beta$ is a prefix
  of $\theta$; by a divsibility argument, each of these periods is
  itself $|\delta|$-periodic, therefore $q$ is $|\delta|$-periodic.
  Besides, observe that $\beta$ is a prefix and $\alpha$ a suffix of
  $r$, therefore $r$ is $|\delta|$-periodic as well.

  Let us show that, for all $n$, we have $\bw(n) = \bw(n+|\delta|)$.
  The only case where this could fail is if $n$ and $n+|\delta|$ do
  not belong to the same occurrence of $q$ nor to the same occurrence
  of $r$ in $\bw$. Since
  $|\delta|=\gcd(\sigma_j-\sigma_1, \sigma_2-\sigma_1) \leq \sigma_2$,
  the only possible span for the $q$-overlap (and the $r$-overlap)
  covering positions $n$ and $n+|\delta|$ is $\sigma_1$.  But then,
  the Figure~\ref{fig:main-fat} shows that the prefix of length
  $|\cv_q(\sigma_1)|-|s_1|$ of $\cv_q(\sigma_1)$ and the prefix of length
  $|\cv_q(\sigma_2)|-|s_2|$ of $\cv_q(\sigma_2)$ are equal; since the latter
  one is $|\delta|$-periodic, the former one is as well. By a length
  argument, $\bw(n) = \bw(n+|\delta|)$ in any case and $\bw$ is
  periodic: a contradiction.
\myqed
\end{proof}

\noindent
Finally, the next lemma gives $2 \iff 3$ in
Proposition~\ref{pro:spec}.

\begin{lem}
  \label{lem:2iff3}
  Let $\bw$ denote a biinfinite word and $q$, $r$ denote two
  quasiperiods of $\bw$ of the same length. The derivated sequences of
  $\bw$ along $q$ and along $r$ are equal if and only if, up to
  swapping $q$ and $r$, there exists a chain of quasiperiods
  $u_1, \dots, u_k$ with $u_1=q$ and $u_k=r$, such that $u_{i+1}$ is
  the successor of $u_i$.
\end{lem}
\begin{proof}
  Let $(q_n)_{n \in \zz}$ and $(r_n)_{n \in \zz}$ denote the sequences
  of \emph{positions} of $q$ and $r$ in $\bw$. If the derivated
  sequences are equal up to a shift of one position, then there exists
  an integer $k$ in $\{1, \dots, |q|-1\}$ such that for all $n$ we
  have $r_n = q_n + k$.  In particular $r$ always starts at the same
  position inside $q$. Differently put, this means that there exists a
  word $s$ of length $k$ such that $r$ is a suffix of $qs$ and each
  occurrence of $q$ in $\bw$ is the prefix of an occurrence of $qs$.
  Set $u_i = (q s)(i, \dots, i+|r|-1)$ and the implication is proved.
  
  Conversely suppose that there is a family of quasiperiods
  $u_0, u_1, \dots, u_{k-1}$ such that $u_0 = q$ and $u_{k-1} = r$ and
  $u_{i+1}$ is the unique successor of $u_i$ for each
  $0 \leq i < k-1$. By Proposition~\ref{pro:predsuc}, none of the
  $u_i$ is right special except maybe $u_{k-1}$. Therefore, there is
  an occurrence of $r$ exactly $k-1$ positions after each occurrence
  of $q$ in $\bw$.  Lemma~\ref{lem:no-qrrq} ensures that no other
  occurrences of $r$ appear in $\bw$, therefore we can conclude that
  $r_n = q_n+k-1$ for each integer $n$.  \myqed
\end{proof}

Finally, the next proposition characterizes compatible and defined
couples of quasiperiods.

\begin{pro}
  Let $q,r$ denote two words of the same length. The couple $(q,r)$ is
  definite if and only if each $q$-quasiperiodic biinfinite word
  contains infinitely many occurrences of $r$.
\end{pro}
\begin{proof}
  If $(q,r)$ is definite, then $f_{q,r}(m,n)$ is defined whenever the
  proper $3$-overlap $w=\cv^*_q(m,n)$ exists; therefore each proper
  overlap of $q$ contains one occurrence of $r$. If a biinfinite word
  is $q$-quasiperiodic, then it contains infinitely many occurrences
  of proper overlaps of $q$, and therefore infinitely many occurrences
  of $r$.
  
  Conversely, suppose that each $q$-quasiperiodic infinite word
  contains infinitely many occurrences of $r$. Suppose that $m$ is an
  integer such that the proper overlap $\cv^*_q(m)$ exists, but does
  not contain an occurrence of $r$. Then the periodic biinfinite word
  given by $\bw = \cv^*_q(\dots m, m, m \dots)$ does not contain any
  occurrence of $r$; but it should also contain infinitely many
  occurrences of $r$ by hypothesis. Therefore we have a contradiction
  and each proper overlap $\cv^*_q(m)$ contains an occurrence of $r$,
  which shows that $(q,r)$ is definite.
\end{proof}

\section{Quasiperiods of biinfinite Sturmian words}
\label{sec:sturm}

If $\bw$ is an infinite word (either indexed by $\nn$ or by $\zz$),
then $P_\bw(n)$ denotes the number of distinct factors of length $n$
in $\bw$ and $P_\bw$ is the \emph{complexity function} of $\bw$. An
infinite word is \emph{Sturmian} if and only if it is not ultimately
periodic and satisfies $P_\bw(n)=n+1$ for each integer
$n$. Equivalently, a word is Sturmian if and only if it is not
eventually periodic and has exactly one right special factor and one
left special factor of each length.  Sturmian words are an important
and well-studied class of infinite words (see~\cite[Chapter
2]{LothaireAlgebraic} and~\cite[Chapter 6]{PytheasFogg2002}). Now we
determine the set of quasiperiods of any biinfinite Sturmian word.  To
this end, if $\bw$ is a biinfinite word, let $Q_\bw(n)$ denote the
number of quasiperiods of length $n$ in $\bw$.

\begin{thm}
  \label{thm:main}
  Let $\bw$ denote a biinfinite Sturmian word and $n$ a nonnegative
  integer.
  \begin{enumerate}
  \item We have $Q_\bw(n) = 0$ if and only if $\bw$ has a nonempty
    bispecial factor of length $n-1$.
  \item If $Q_\bw(n)>0$ and $s$ denotes the shortest bispecial factor
    of $\bw$ with $|s|\geq n$, then quasiperiods of length $n$ in
    $\bw$ are exactly the factors of length $n$ in $s$.
  \end{enumerate}
\end{thm}

\begin{proof}
  We prove the two statements separately.

  \emph{Statement~1}. If $\bw$ has a nonempty bispecial factor
  of length $n-1$, say $u$, then there exists a letter $\alpha$ such
  that the (unique) right special factor of length $n$ in $\bw$ is
  $\alpha{}u$, because any suffix of a right special factor is also
  right special. By Theorem~\ref{thm:old} the word $\alpha{}u$ is not
  a quasiperiod of $\bw$; by Proposition~\ref{pro:predsuc} if $\bw$
  had a quasiperiod of length $n$, then it would have a right special
  quasiperiod of this length. Consequently $\bw$ has no quasiperiod of
  length $n$.

  Conversely suppose that $\bw$ has no bispecial factor of length
  $n-1$. Since $\bw$ is Sturmian, it has exactly one right-special and
  one left-special factor of length $n$, so its set of factors of
  length $n$ may be written
  $\{c_0, \dots, c_{k-1}\} \cup \{d_0, \dots, d_{\ell-1}\} \cup \{e_0,
  \dots, e_{m-1}\}$,
  where $c_{k-1}$ is right special and has successors $d_0$ and $e_0$;
  both $d_{\ell-1}$ and $e_{m-1}$ have successor $c_0$, which is left
  special; each other $c_i$, $d_i$ and $e_i$ has respectively
  $c_{i+1}$, $d_{i+1}$ and $e_{i+1}$ as an (unique) successor.  We
  have $k \geq 1$, but we might have $m=0$ or $\ell=0$.
  Figure~\ref{fig:graphsturm} shows a graph of the ``successor''
  relation. Observe that $k+\ell+m = n+1$ and $k > 2$ (otherwise, we
  would have a bispecial factor of length $n-1$). As a consequence,
  the maximal distance between two consecutive occurrences of $c_0$ is
  $\max(k+\ell, k+m)$, which is bounded by $n$. In other terms, $c_0$
  is a quasiperiod of $\bw$.

  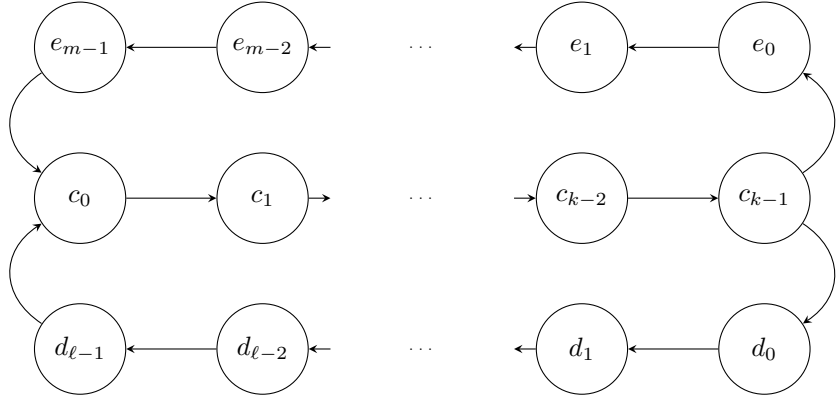
\begin{figure}[htbp]
    \centering
    \begin{tikzpicture}[xscale=3,yscale=2,minimum size=1.2cm,>=stealth]
      \node[draw,circle] (q1) at (0,0) {$c_0$};
      \node[draw,circle] (q2) at (0.8,0) {$c_1$};
      \node (qint1) at (1.3,0) {};
      \node (qdots) at (1.5,0) {{\tiny $\dots$}};
      \node (qint2) at (1.7,0) {};
      \node[draw,circle] (q3) at (2.2,0) {$c_{k-2}$};
      \node[draw,circle] (q4) at (3,0) {$c_{k-1}$};

      \path[draw,->] (q1) to (q2);
      \path[draw,->] (q2) to (qint1);
      \path[draw,->] (qint2) to (q3);
      \path[draw,->] (q3) to (q4);
      
      \node[draw,circle] (t1) at (3,-1) {$d_0$};
      \node[draw,circle] (t2) at (2.2,-1) {$d_1$};
      \node (tint1) at (1.7,-1) {};
      \node (tdots) at (1.5,-1) {{\tiny $\dots$}};
      \node (tint2) at (1.3,-1) {};
      \node[draw,circle] (t3) at (0.8,-1) {$d_{\ell-2}$};
      \node[draw,circle] (t4) at (0,-1) {$d_{\ell-1}$};

      \path[draw,->] (q4) edge[bend left=45] (t1);
      \path[draw,->] (t1) to (t2);
      \path[draw,->] (t2) to (tint1);
      \path[draw,->] (tint2) to (t3);
      \path[draw,->] (t3) to (t4);
      \path[draw,->] (t4) edge[bend left=45] (q1);

      \node[draw,circle] (b1) at (3,+1) {$e_0$};
      \node[draw,circle] (b2) at (2.2,+1) {$e_1$};
      \node (bint1) at (1.7,+1) {};
      \node (bdots) at (1.5,+1) {{\tiny $\dots$}};
      \node (bint2) at (1.3,+1) {};  
      \node[draw,circle] (b3) at (0.8,+1) {$e_{m-2}$};
      \node[draw,circle] (b4) at (0,+1) {$e_{m-1}$};

      \path[draw,->] (q4) edge[bend right=45] (b1);
      \path[draw,->] (b1) to (b2);
      \path[draw,->] (b2) to (bint1);
      \path[draw,->] (bint2) to (b3);
      \path[draw,->] (b3) to (b4);
      \path[draw,->] (b4) edge[bend right=45] (q1);
    \end{tikzpicture}
    \caption[Rauzy graph of a Sturmian word]{Successor graph (usually
      called \emph{Rauzy graph}) of factors of length $n$ of a
      Sturmian word}
    \label{fig:graphsturm}
  \end{figure}

  \emph{Statement~2}. Let $\bw$ denote a biinfinite Sturmian
  word and suppose that $Q_\bw(n)>0$ for some integer $n$. The left
  special and the right special factors of length $n$ of $\bw$, call
  them $\ell$ and $r$, are both quasiperiods by
  Proposition~\ref{pro:predsuc}. Call $s$ the shortest factor of $\bw$
  having $\ell$ as a prefix and $r$ as a suffix. The set of factors of
  length $n$ of $s$ is given by a sequence $u_1, u_2, \dots u_k$,
  where $u_1 = \ell$ and $u_k = r$, such that $u_{i+1}$ is the
  successor of $u_i$ for each $1 \leq i < k$. By
  Proposition~\ref{pro:predsuc} again, the set $\{u_1, \dots, u_k\}$
  is the set of quasiperiods of length $n$ of $\bw$. Observe that $s$
  is, by definition, exactly the shortest bispecial factor of $\bw$
  not shorter than $n$.
\end{proof}

Since any Sturmian word has infinitely many bispecial factors, whose
difference between consecutive lengths are unbounded, we have:

\begin{cor}
  Each biinfinite Sturmian word $\bw$ has infinitely many quasiperiods.
  Moreover, $Q_\bw$ is unbounded.
\end{cor}

\section{Conclusion}
\label{sec:conclu}

As explained in the introduction, the biinfinite case may give nicer results about
quasiperiodicity of subshifts and of Sturmian words. This paper provided a toolbox to
study the quasiperiods of biinfinite words, but many questions are still to be answered.

\begin{enumerate}
\item An $\nn$-word $\bw$ is periodic if and only if $Q_\bw(n)>0$ for each large enough
  $n$. Is it possible to characterize ultimately periodic $\zz$-words in terms of
  quasiperiods?
\item An $\nn$-word $\bw$ is standard Sturmian if and only if it satisfies $Q_\bw(n)=0$
  exactly when there is a bispecial factor of length $n-1$ in $\bw$. Is there a
  characterization of biinfinite Sturmian words in terms of quasiperiods?
\item What about other families of low-complexity sequences, such as episturmian or
  Arnoux-Rauzy sequences?
\item If an $\nn$-word is multi-scale quasiperiodic, then it is
  uniformly recurrent~\cite{MarcusMonteil2006Arxiv}. It is easy to
  construct a $\zz$-word which is multi-scale quasiperiodic but not
  uniformly recurrent: the word ${}^\omega(ba)\cdot(ab)^\omega$ has
  quasiperiod $a(ba)^n$ for each positive $n$, but the factor $aa$
  occurs only once. Are all such words ultimately periodic? If not,
  how can they be characterized?
\item If a biinfinite word has infinitely many quasiperiods, does it necessarily have
  quasiperiodic derivated sequences? If not, can we build a counter-example which is
  uniformly recurrent?
\end{enumerate}

\paragraph{Acknowledgements.}
The authors would like to thank Gwena{\"e}l Richomme for the proof of
Lemma~\ref{lem:no-qrrq} and for proofreading many versions of this paper,
which notably helped to clarify the statement of Theorem~\ref{thm:main}.
The authors are also grateful to Patrice S{\'e}{\'e}bold, who provided
Lemma~\ref{lem:patrice} and its proof.

\bibliographystyle{plain}
\bibliography{bisturm}

\end{document}